\documentclass[reqno,12pt]{article}

\usepackage{amsmath, amsgen,amstext,amsbsy,amsopn,amsthm, amssymb}
\newcommand{\version}{January 10, 2014}
\usepackage{graphicx, epstopdf}

\usepackage{latexsym}
\newcommand{\Bb}{\mathbb{R}}

\newcommand{\RR}{\mathbb{R}}

\renewcommand{\lor}{\Lambda}     
\newcommand{\Boo}[2]{\lor_{#1}(#2)}         

\newcommand{\Ibb}[1]{ {\rm I\ifmmode\mkern -3.6mu\else\kern -.2em\fi#1}}
\newcommand{\ibb}[1]{\leavevmode\hbox{\kern.3em\vrule
     height 1.2ex depth -.3ex width .2pt\kern-.3em\rm#1}}

\newcommand{\OO}{\mathcal O}

\newcommand{\F}{\mathcal F}

\newtheorem{thm}{Theorem}

\theoremstyle{definition}

\newcommand{\be}{\begin{equation}}
\newcommand{\ee}{\end{equation}}
\newcommand{\beq}{\begin{equation}}
\newcommand{\eeq}{\end{equation}}
\newcommand{\bea}{\begin{eqnarray}}
\newcommand{\eea}{\end{eqnarray}}
\newcommand{\beqa}{\begin{eqnarray}}
\newcommand{\eeqa}{\end{eqnarray}}

\newcommand{\eps}{\varepsilon}

\newcommand{\half}{\mbox{$\frac{1}{2}$}}

\newcommand{\C}{\mathbb{C}}

\newcommand{\h}{\mathcal H}

\newcommand{\bh}{\mathcal B({\mathcal H})}

\newcommand{\m}{\mathcal A}
\newcommand\mo{\m(\mathcal O)}

\newcommand{\bho}{\mathcal B({\mathcal H}_{1})}
\newcommand{\bht}{\mathcal B({\mathcal H}_{2})}

\def\mfr#1/#2{\hbox{$\frac{{#1} }{ {#2}}$}}

\begin{document}
  \markboth{\scriptsize{JY \version}}{\scriptsize{JY \version}}
    \title{Localization and Entanglement\\ in Relativistic Quantum Physics}
   
\author{Jakob Yngvason\\ \normalsize\it Faculty of Physics, University of Vienna,\\ \normalsize\it 
Boltzmanngasse 5, 1090
Vienna, Austria\\
{\small\rm Email: jakob.yngvason@univie.ac.at}} 
\date{\version}
\maketitle

\sloppy
        

\section{Introduction}These notes are a slightly expanded version of a  lecture presented in February 2012 at the workshop
``The Message of Quantum Science -- Attempts Towards a Synthesis" held at the ZIF in Bielefeld. The participants were physicists with a wide range of different expertise and interests. The lecture was intended as a survey of a {\it small selection} of  the insights into the structure of relativistic quantum physics that have accumulated through the efforts of many people over more than 50 years\footnote{including, among many others, R. Haag, H. Araki, D. Kastler, H.-.J,. Borchers, A. Wightman, R. Streater, B. Schroer, H. Reeh, S. Schlieder, S. Doplicher,  J. Roberts,   R. Jost, K. Hepp,  J. Fr\"ohlich, J. Glimm, A. Jaffe, J. Bisognano, E. Wichmann, D. Buchholz, K. Fredenhagen, R. Longo, D. Guido, R. Brunetti, J. Mund, S. Summers, R. Werner,  H. Narnhofer, R.\ Verch, G.\ Lechner,\dots}. This contribution discusses some facts about relativistic quantum physics, most of which are quite familiar to practitioners of {\it Algebraic Quantum Field Theory} (AQFT)\footnote{also known as {\it Local Quantum Physics} \cite{Haag}.} but less well known outside this community. No claim of originality is made; the goal of this contribution is merely to present these facts in a simple and concise manner, focusing on the following issues:

\begin{itemize}
\item Explaining how quantum mechanics (QM) combined with (special) relativity, in particular an upper bound on the propagation velocity of effects,  leads naturally to systems with an infinite number of degrees of freedom (relativistic quantum fields).

 \item A brief summary of the differences in mathematical structure compared to the QM of finitely many particles that emerge form the synthesis with relativity, in particular different localization concepts, type III von Neumann algebras rather than type I, and
 ``deeply entrenched" \cite{Clifton} entanglement, 
\item Comments on the question whether these mathematical differences have significant consequences for the physical interpretation of basic concepts of QM.

\end{itemize}

\section{What is Relativistic Quantum Physics?}

According to E. Wigner's groundbreaking analysis from 1939 of relativistic symmetries in the quantum context  \cite{Wig} any relativistic quantum theory should contain as {\it minimal ingredients}
\begin{itemize} 
\item A Hilbert space $\mathcal H$ of state vectors.
\item A unitary representation $U(a,\Lambda)$ of the inhomogeneous (proper, orthochronous) Lorentz group (Poincar\'e group) $\mathcal P^\uparrow_+$ on $\mathcal H$.\footnote{More precisely, also representations ``up to a phase" are allowed, which amounts to replacing $\mathcal P^\uparrow_+$ by its universal covering group $ISL(2,\mathbb C)$.} Here $a\in\mathbb R^4$ denotes a translation of Minkowski space and $\Lambda$ a Lorentz transformation.\footnote{For simplicity of the exposition we refrain from discussing the possibility that the Lorentz transformations act only as automorphisms on the algebra of observables but are not unitarily implemented on the Hilbert space of states under consideration, as can be expected in charged superselection sectors of  theories with massless particles \cite{Froehlich, Bsymmbreak}.}
\end{itemize}
The representations were completely  classified by Wigner in \cite{Wig}. 

 The representation $U(a,{\bf 1})=: U(a)$ of the translations leads directly to the observables  {\it energy} $P^0$ and  {\it momentum} $P^{\mu}$, $\mu=1,2, 3$ as the corresponding infinitesimal generators\footnote{Here, and in the following, units are chosen so that Planck's constant, $\hbar$, and the velocity of light, $c$, are equal to 1. The metric on Minkowski space is $g_{\mu\nu}=\text{diag}\,(1,-1,-1,-1)$.}: 
 \beq   U(a)=\exp\left(\mathrm i \sum_{\mu=0}^3 P^\mu a_\mu\right).\eeq
The {\it stability requirement} that the  energy operator, $P^0$, should be bounded below 
implies that the joint spectrum of the commuting operators $P^\mu$ is contained in the forward light cone $\mathrm V^+$.  This is called the relativistic {\it spectrum condition}. The operator of the mass  is 
$M=(\sum_{\mu=0}^3P^\mu P_\mu)^{1/2}$. In an irreducible representations of $\mathcal P^\uparrow_+$ fulfilling the spectrum condition the mass has a sharp value $m\geq 0$. These representations  fall into three classes:

1. The massive representations $[m,s]$ with the mass $m>0$ and the spin $s=0,\half, 1,\dots$ labeling the irreducible representations of the stabilizer group $SU(2)$ of the energy-momentum vector in the rest frame.

2. The massless representations $[0,h]$ of finite helicity $h=0,\pm\half, \pm1,\dots$ corresponding to one dimensional representations of the stabilizer group of a light-like energy-momentum vector (the two-dimensional euclidian group $E_2$).

3. Massless representations of unbounded helicity (``infinite spin representations") corresponding to infinite dimensional representations of the stabilizer group $E_2$.\\

Besides energy and momentum, further observables, namely the angular momentum operators and the generators of Lorentz boosts,
follow from the representation $U(0,\Lambda)$ of the homogeneous Lorentz group. In Section \ref{modularloc} below we shall see how an intrinsic localization concept (``modular localization") can be associated with the latter, but first we discuss the problems that arise if one tries to mimic the procedure in non-relativistic QM and define localization via position operators for particles.


\subsection{Problems with Position Operators}
In non-relativistic quantum mechanics, spatial {\it localization} of state vectors is determined through the spectral projectors of position operators. For instance, a single particle state with wave functions $\psi(\mathbf x)$ is localized in a domain $\Delta\subset \mathbb R^3$ if and only if the support of $\psi$ lies in $\Delta$, which means that $E_\Delta\psi=\psi$ with $E_\Delta$  the multiplication operator by the characteristic function of $\Delta$. Time evolution generated by the non-relativistic Hamiltonian $H=\frac 1{2m}\mathbf P^2=-\frac1{2m}\nabla^2$ immediately spreads out the localization in the sense that for any pair $\Delta, \Delta'$ of disjoint domains so that $E_\Delta E_{\Delta'}=0$ we have $\exp(itH)E_\Delta\exp(-itH)E_{\Delta'}\neq 0$ for all $t\neq 0$, no matter how far $\Delta$ and $\Delta'$ are from each other. Since there is no upper bound to the velocity of propagation of effects in non-relativistic QM this is not a surprise. In a relativistic theory, on the other hand, for instance with $H=(c^2 \mathbf  P^2+m^2c^4)^{1/2}$, one might expect that $\exp(itH)E_\Delta\exp(-itH)E_{\Delta'}$ stays zero as long as $c|t|$ is smaller than the spatial distance between the two domains.
This, however, is {\it not} the case, due to the analyticity implied by the relativistic spectrum condition:\footnote{In this form the result was first published by J.F. Perez and I.F. Wilde in \cite{PerezWilde}. See also \cite{Malament} for the same conclusion under slightly weaker premises.}

\begin{thm}[{\bf Localization via position operators is in conflict with causality}]Suppose there is a mapping $\Delta\mapsto E_\Delta$ from subsets of space-like hyperplanes in 
Minkowski space into projectors on $\mathcal H$ such that
\begin{itemize}
\item[(1)] $U(a)E_\Delta U(a)^{-1}=E_{\Delta+a}$
\item[(2)] $E_\Delta E_{\Delta'}=0$ if $\Delta, \Delta'$ space-like separated. 
\end{itemize}
Then $E_\Delta=0$ for all $\Delta$.\end{thm}

\begin{proof} The spectrum condition implies that the function  $a\mapsto U(a)\Psi$ has, for every $\Psi\in\mathcal H$,  an {\it analytic continuation} into $\mathbb R^4+\mathrm i {\rm V}^+\subset \mathbb C^4$. The second condition (2) means that $\langle E_\Delta \Psi, U(a)E_\Delta \Psi\rangle=\langle \Psi, E_\Delta E_{\Delta+a}U(a)\Psi\rangle=0$ on an open set in Minkowski space. But an analytic function that is continuous on the real boundary of its analyticity domain  and vanishing on an open subset of this boundary vanishes identically\footnote{This follows from the ``edge of the wedge" theorem, that is a generalization of the Schwarz reflection principle to several complex variables, see, e.g., \cite{SW}.}.\end{proof}
The conclusion that can be drawn from this result is that localization in terms of position operators is {\it incompatible with causality} in relativistic quantum physics.\footnote{This objection does not exclude {\it approximate} localization in the sense of Newton and Wigner \cite{NW}.}

This dilemma is resolved by realizing that the relevant concept in relativistic quantum physics is localization of {\it quantum fields} rather than localization of wave functions of particles in position space. The space-time points $x$ appear as variables of the quantum field operators $\Phi_{\alpha}(x)$ (with $\alpha$ a tensor or spinor index). Causality manifests itself in {\it commutativity} (or anticommutativity) of the  operators at space-like separation of the variables.

Taken together, covariance w.r.t. space-time translations, the spectrum condition and local (anti)commutativity imply that the dependence of field operators on the coordinates is by necessity singular, and well defined operators are only obtained after smearing with test functions.
This means that quantum fields are operator valued distributions rather than functions  \cite{SW, GaardW}, i.e., only ``smeared" operators $\Phi_\alpha(f)$ with $f$ a test function on space-time are well defined.   Localization of field operators {\it at a point} is thus a somewhat problematic concept\footnote{Field operators at a point can, however,  be  defined as {\it quadratic forms} on vectors with sufficiently nice high energy behavior.}, while localization in a domain of space-time (the support of the test function $f$) has a clear meaning.

These ideas are incorporated in the general conceptual framework of  {\it Algebraic Quantum Field Theory} (AQFT), also called {\it Local Quantum Physics} (LQP) \cite{HK, Haag, Arakibook, BuchhHaag}. Here the emphasis is on the collection (``net") of operator algebras generated by quantum field operators localized in different domains of space-time. The quantum fields themselves appear as auxiliary objects since many different quantum fields can generate the same net of algebras. The choice of some definite field to describe a given net is somewhat analogous to the choice of a coordinate system in differential geometry. In some cases the net is even defined without any reference to quantum fields in the traditional sense, and important general results of the theory do not, in fact,  rely on a description of the net in terms of operator valued distributions.


\section{ Local Quantum Physics}
\subsection{The general assumptions}

The basic ingredients of a model in LQP are:
\begin{itemize} 
\item A separable Hilbert space $\mathcal H$ of state vectors.
\item A Unitary representation $U(a,\Lambda)$ of the Poincar\'e group $\mathcal P^\uparrow_+$ on $\mathcal H$.\footnote{More generally, a representation of the covering group $ISL(2,\mathbb C)$.}
\item An invariant, normalized state vector $\Omega\in \mathcal H$ (vacuum), unique up to a phase factor.
\item  A family of *-algebras $\mathcal F(\OO)$ of operators\footnote{For mathematical convenience  we assume that the operators are bounded and that the algebras are closed in the weak operator topology, i.e., that they are {\it von Neumann algebras}. The generation of such algebras from unbounded  quantum field operators $\Phi_\alpha(f)$ is in general a nontrivial issue that is dealt with, e.g., in \cite{Driessleretal, BorchY1}. In cases when the real and imaginary parts of the field operators are essentially self-adjoint, one may think of the $\mathcal F(\OO)$ as generated by bounded functions (e.g., spectral projectors, resolvents, or exponentials)  of these operators smeared with test functions having support in $\mathcal O$. More generally, the polar decomposition of the unbounded operators can be taken as a starting point for generating the local net of von Neumann algebras.} on
$\mathcal H$ (a ``field net"), indexed by regions $\OO\subset \mathbb R^4$ with $\mathcal F(\OO_1)\subset \mathcal F(\OO_2)$ if $\OO_1\subset\OO_2 (isotony).$\end{itemize}

The requirements are:
\begin{itemize}
 \item [R1.] {\it Local (anti-)commutativity\/}: $\mathcal F(\OO_1)$ commutes with $\mathcal F(\OO_2)$ if $\OO_1$ and $\OO_2$ space-like separated. (Or, in the case of Fermi statistics,  commutes after a ``twist" of the fermionic parts of the algebras, cf. \cite{BisWich2}, Eq. 33b.)
 \item  [R2.] {\it Covariance\/}: $U(a,\Lambda)\F(\OO)U(a,\Lambda)^{-1}=\F(\Lambda\OO+a)$.
 \item [R3.]{\it Spectrum condition\/}: The energy momentum spectrum, i.e., the joint spectrum of the generators of the translations $U(a)$ lies in  ${\rm V}^+$.
 \item [R4.] {\it Cyclicity of the vacuum\/}:  $\cup_\OO\F(\OO)\Omega$ is dense in $\h$.
\end{itemize}


\noindent{\bf Remarks:}

\begin{itemize}
\item The operators in $\mathcal F(\OO)$ can intuitively be thought of as generating physical operations that can be carried out in the space-time region $\OO$. 

\item Usually (but not always!) $\mathcal F(\OO)$ is nontrivial for all open regions $\OO$.

\item Associated with the field net $\{\F(\OO)\}_{\OO\subset \mathbb R^4}$ there is usually another net of operator algebras, $\{\m(\OO)\}_{\OO\subset \mathbb R^4}$,  representing local {\it observables} and commuting  with the field net and itself at space-like separations. Usually this is a {\it subnet} of the field net, selected by invariance under some (global) gauge group.\footnote{In the theory of superselection sectors, initiated by H.J. Borchers in \cite{BorchSuper} and further developed in particular by S. Doplicher, R. Haag and J.E. Roberts in \cite{DHR1, DHR2, DHR3, DHR4}, the starting point is the net of observables while the field net and the gauge group are derived objects. For a very recent development, applicable to theories with long range forces, see \cite{BuchhRob}.}
\end{itemize}

Thus, the operators in $\mathcal F(\OO)$ can have two roles:
\begin{itemize}
\item[1.] They implement local transformations of states\footnote{Here and in the sequel, a {\it state} means a  positive, normalized linear functional on the algebra in question, i.e., a linear functional such that $\omega(A^*A)\geq 0$ for all  $A$ and 
$\omega({\mathbf 1})=1$. We shall also restrict the attention to {\it normal} states, i.e., $\omega(A)={\rm trace\ } (\rho A)$ with a nonnegative trace class operator $\rho$ on $\mathcal H$ with trace 1.}
$\omega$ in the sense of K. Kraus \cite{Kraus}, i.e., 
\beq  \omega\mapsto \sum_i \omega(K_i^*\,\cdot\,K_i)\eeq
with $K_i\in \mathcal F(\OO)$. 
\item[2.] The self-adjoint elements of $\m(\OO)$ correspond to physical properties of the system that can, at least in principle,  be measured in $\OO$.
\end{itemize}
Already in the mid 1950's  Rudolf Haag \cite{Haagthm, Haagorg} had the fundamental insight that information about interactions between {\em particles}, that emerge asymptotically for large positive or negative instances of time but are usually not unambiguously defined at finite times, is already {\it encoded in the field net}. In order to determine the particle spectrum of a given theory and compute scattering amplitudes it is not necessary to attach specific interpretations to specific operators in $\F(\OO)$ besides their localization.

\subsection{Construction Methods}

Traditionally,  the main methods to construct models in quantum field theory have been:

\begin{itemize}
\item Lagrangian field theory plus canonical quantization. This leads rigorously to free fields and variants like generalized free fields, Wick-powers of such fields etc. that satisfy the Wightman-G\aa rding axioms \cite{SW, GaardW, Jost, BogoliubovBook}. Perturbation theory plus renormalization leads  (also rigorously!) to theories with {\it interactions} defined in terms of {\it formal power series} in a coupling constant. (See \cite{BDF} for a modern, rigorous version of perturbation theory for quantum fields.)
\item Constructive QFT (J.\ Glimm, A.\ Jaffe and others) \cite{GlimmJaff}.  Here the renormalization of certain lagrangian field theories is carried out rigorously, without recourse to perturbation theory.  In this way,  models of interacting fields in space-time dimensions 1+1 and 1+2 have been obtained,  {\it but so far not in 1+3 dimensions}.
\item Conformal QFT in 1+1 space-time dimensions based on Virasoro algebras and other algebraic structures, see, e.g., \cite{KawahigashiLongo} and references cited therein.
\end{itemize}
It is a big challenge in QFT to develop new methods of construction and classification. Recently, progress has been achieved through {\it deformations} of known models \cite{GrosseLechner, BLS}. In particular, G. Lechner has shown that a large class of integrable models in 1+1 dimensions, and many more, can be obtained from deformations of free fields \cite{Lechner}. See also \cite{Alazzawi}. There is also a very recent approach due to J. Barata, C. J\"akel and J. Mund \cite{BJM} based on Tomita-Takesaki modular theory \cite{Tomita, BR}. The latter concerns operator algebras with a cyclic and separating vector, and applies to the local algebras of relativistic quantum field theory because of the Reeh-Schlieder theorem discussed next.

\subsection{The Reeh-Schlieder Theorem}

The Reeh-Schlieder Theorem was originally derived in the context of Wighman quantum field theory in \cite{ReehSchlied}.  For general nets of local algebras as in Sect. 3.1, one additional assumption is needed, {\it weak additivity\/}: For every fixed open set $\mathcal \OO_0$ the algebra generated by the union of all translates, $\F(\OO_0+x)$, is dense in the union of all $\F(\OO)$ in the weak operator topology. If the net is generated by Wightman fields this condition is automatically fulfilled.

\begin{thm}[{\bf Reeh-Schlieder}]
Under the assumption of week additivity, $\F(\OO)\Omega$ is dense in $\mathcal H$ for all open sets $\OO$, i.e., {\it the vacuum is cyclic for every single local algebra} and not just for their union.\end{thm}
 
\begin{proof} Write $U(a)$ for $U(a,\bf 1)$. Pick $\mathcal O_0\subset \OO$ such that $\mathcal O_0+x\subset \OO$ for all $x$ with $|x|<\varepsilon$, for some $\eps>0$. If {$\Psi \perp \F(\OO)\Omega$}, then {$\langle \Psi,U(x_1)A_1U(x_2-x_1)\cdots U(x_n-x_{n-1})A_n\Omega\rangle=0$} for all $A_i\in\mathcal O_0$ and {$|x_i|<\eps$}. Then use the analyticity of $U(a)$ to conclude that this must hold for {\it all} {$x_i$}. The theorem now follows by appealing to weak additivity. \end{proof}

\noindent{{\bf Corollary.}} {\it The vacuum is a {\it separating} vector of $\F(\OO)$ for every $\OO$ such that its causal complement $\OO'$ has interior points, i.e.,
$A\Omega=0$ for $A\in\F(\OO)$ {\it implies} {$A=0$}. Moreover, if $A$ is a positive operator in $\F(\OO)$ and $\langle\Omega, A\Omega\rangle=0$, then $A=0$.}

\begin{proof} If $\mathcal O_0\subset \OO'$, then $AB\Omega=BA\Omega=0$ for all $B\in\F(\OO_0)$\footnote{For simplicity we have assumed local commutativity. In the case of Fermi fields the same conclusion is drawn by splitting the operators into their bosonic and fermonic parts.}. But $\F(\OO_0)\Omega$ is dense if $\OO_0$ is open, so $A=0$. The last statement follows because the square root of a positive $A\in\F(\OO)$ belongs also to $\F(\OO)$. \end{proof}

\noindent{\bf Remarks}
\begin{itemize}
\item The Reeh-Schlieder Theorem and its Corollary hold, in fact, not only for the vacuum vector $\Omega$ but for any state vector $\psi$ that is an {\em analytic vector for the energy}, i.e. such that  $e^{\mathrm i x_0 P^0}\psi$ is an analytic function of $x_0$ in a whole complex neighborhood of 0. This holds in particular if $\psi$ has bounded energy spectrum. 

\item  {\it No violation of causality} is implied by the Reeh-Schlieder Theorem. The theorem is ``just" a manifestation of unavoidable {\it correlations} in the vacuum state (or any other state given by an analytic vector for the energy) in relativistic quantum field theory. (See the discussion of  entanglement in Sect. 5.)

\item Due to the {\it cluster property}, that is a consequence of the uniqueness of the vacuum, the correlation function
\beq  F(x)=\langle\Omega, AU(x)B\Omega\rangle-\langle\Omega, A\Omega\rangle \langle\Omega, B\Omega\rangle\eeq
for two local operators $A$ and $B$ tends to zero if $x$ tends to space-like infinity. If there is a mass gap in the energy momentum spectrum the convergence is exponentially fast \cite{Fredclust}. Thus, although the vacuum cannot be a product state for space-like separated local algebras due to the Reeh-Schlieder Theorem, the correlations become very small as soon as the space-like distance exceeds the Compton wavelength associated with the mass gap \cite{WS, Zych}.\end{itemize}

\subsection{Modular structures and the Bisognano-Wichmann Theorem}

A remarkable development in the theory of operator algebras was initiated 1970 when M. Takesaki published his account \cite{Takesaki} of M. Tomita's theory of modular Hilbert algebras developed 1957-67. In 1967 similar structures had  independently been found by R. Haag, N. Hugenholz and M. Winnink in their study of thermodynamic equilibrium states of infinite systems \cite{HHW},
and in the 1970's the theory found its way into LQP. Various aspects of these developments are  discussed in the review article \cite{BorchMod}, see also \cite{SummMod} for a concise account. On the mathematical side Tomita-Takesaki theory is the basis of A. Connes' groundbreaking work on the classification von Neumann algebras \cite{Connes}.

The Tomita-Takesaki modular theory concerns a von Neumann algebra $\mathcal A$ together with a cyclic and separating vector $\Omega$.  To every such pair it associates a one parameter group of unitaries  (the {\it modular group}) whose adjoint action leaves the algebra invariant, as well as an anti-unitary involution (the {\em modular conjugation})  that maps the algebra into its commutant $\mathcal A'$. The precise definition of these objects is as follows.

First, one defines an antilinear operator $S_0$: $\mathcal A\Omega\to\mathcal A\Omega$ by
\beq  S_0A\Omega=A^*\Omega.\eeq
This  operator (in general unbounded) is well defined on a dense set in $\mathcal H$ because $\Omega$ is separating and cyclic. We denote by $S$ its closure, $S=S_0^{**}$. It has a polar decomposition
\beq  S=J\Delta^{1/2}=\Delta^{-1/2}J\eeq
with the {\it modular operator} $\Delta=S^*S>0$ and the anti-unitary {\em modular conjugation} $J$  with $J^2={\mathbf 1}$.

The basic facts about these operators are stated in the following Theorem. See, e.g.,  \cite{BR} or \cite{Takesaki} for a proof.

\begin{thm}[{\bf Modular group and conjugation; KMS condition}]
\beq  \Delta^{\mathrm i t}\mathcal A\Delta^{-\mathrm i t}=\mathcal A\quad\text{for all $t\in \mathbb R$}, \quad J\mathcal AJ=\mathcal A'\ .\eeq
Moreover, for $A,B\in \mathcal A$,\footnote{Eq. (7) is, strictly speaking, only claimed for $A,B$ in a the dense subalgebra of ``smooth" elements of $\mathcal A$ obtained by integrating $\Delta^{\mathrm i t}A\Delta^{-\mathrm i t}$ with a test functions in $t$.}

\beq  \label{kms} \langle\Omega, AB\Omega\rangle=
\langle\Omega, B\Delta^{-1}A\Omega\rangle.\eeq
\end{thm}
Eq. \eqref{kms} is equivalent to the Kubo-Martin Schwinger (KMS) condition that characterizes thermal equilibrium states with respect to the ``time" evolution $A\mapsto \alpha_t(A):=\Delta^{\mathrm i t}A\Delta^{-\mathrm i t}$ on $\mathcal A$ \cite{HHW}.\footnote{Due to a sign convention in modular theory the temperature is formally $-1$, but by a scaling of the parameter $t$, including an invertion of the sign,  can produce any value of the temperature.} 

Most of the applications of modular theory to LQP rely on the fact that the modular group and conjugation for an algebra corresponding to a space-like wedge in Minkowski space and the vacuum have a geometric interpretation.
A {\em space-like wedge} $\mathcal W$ is, by definition, a Poincar\'e transform of the {\it standard wedge}
\beq   \mathcal W_{1}=\{x\in\RR^4:\, |x_0|<x_1\}.\eeq 
With $\mathcal W$ is associated a one-parameter family $\Lambda_{\mathcal W}(s)$ 
of Lorentz boosts that leave $\mathcal W$ invariant. The boosts for the standard wedge are in the $x_0$-$x_1$ plane given by the matrices
\beq  \Lambda_{\mathcal W_1}(s)=\begin{pmatrix} \cosh s & \sinh s \\ \sinh s & \cosh s  \end{pmatrix}.\eeq
 There is also a reflection, $j_{\mathcal W}$, 
about the edge of the wedge that maps $\mathcal W$ into the opposite wedge (causal complement) $\mathcal W'$.
For the standard wedge $\mathcal W_1$ the reflection is the product of the space-time inversion {$\theta$} and a rotation $R(\pi)$ by $\pi$ around the 1-axis. For a general wedge the transformations $\Lambda_{\mathcal W}(s)$ and $j_{\mathcal W}$ are obtained from those for $\mathcal W_1$  by combining the latter with a Poincar\'e transformation that takes $\mathcal W_1$ to $\mathcal W$.

Consider now the algebras $\mathcal F(\mathcal W)$ of a field net generated by a Wightman quantum field with the vacuum $\Omega$ as cyclic and separating vector. For simplicity we consider the case 
of Bose fields.\footnote{Fermi fields can be included by means of a ``twist" that turns anticommutators into commutators as in \cite{BisWich2}.} 
The modular objects $\Delta$ and $J$ associated with $(\F(\mathcal W),\Omega)$ depend on $\mathcal W$ but it is sufficient to consider $\mathcal W_1$. As discovered by J. Bisognano and E. Wichmann in 1975 \cite{BisWich1, BisWich2} $\Delta$ and $J$ are related to the representation $U$ of the Lorentz group and the PCT operator $\Theta$ in the following way:
\begin{thm}[{\bf Bisognano-Wichmann}]
\beq  J=\Theta U(R(\pi))\quad \text{\rm and} \quad\Delta^{{\rm i} t}=U(\Lambda_{\mathcal W_1}(2\pi t)).\eeq
\end{thm}


\subsection{Modular Localization 
}\label{modularloc}
 {\it Modular 
localization} \cite{BGL} is based on a certain converse of the Bisognano-Wichmann theorem. This concept associates a localization structure with any (anti-)unitary representation of the proper Poincar\'e group $\mathcal P_+$ (i.e., $\mathcal P_+^\uparrow$ augmented by space-time reflection) satisfying the spectrum condition, in particular the one-particle representations. Weyl quantization then generates naturally a local net satisfying all the requirements (1--4) in Sect. 3.1, including commutativity (or anti-commutativity) at space-like separation of the localization domains. A sketch of this constructions is given in this subsection, focusing for simplicity on the case of local commutativity rather than anti-commutativity.

Let $U$ be an (anti-)unitary representation of ${\mathcal P}_{+}$ 
satisfying the spectrum condition on a 
Hilbert space ${\mathcal H}_1$.  {}
For a given space-like wedge $\mathcal W$, let $\Delta_\mathcal W$ be the unique positive operator
satisfying
\beq     
\Delta_\mathcal W^{\mathrm it}=U(\Boo{\mathcal W}{2\pi t})\;, \quad t\in\Bb\,, 
\eeq  and let 
$J_{\mathcal W}$ to be the anti--unitary involution representing $j_{\mathcal W}$. We define
\beq   S_\mathcal W:=J_\mathcal W\,\Delta_\mathcal W^{1/2}.\eeq


The space
\beq   
{\mathcal K}(\mathcal W):= \{\phi\in{\rm domain} \,\Delta_\mathcal W^{1/2}:\; S_\mathcal W\phi=\phi \}\subset{\mathcal H}_1 
\eeq 
satisfies:
\begin{itemize}
\item ${\mathcal K}(\mathcal W)$ is a closed {\it real}
subspace  of ${\mathcal H}_1$ in the real scalar product ${\rm 
Re}\,\langle\cdot,\cdot\rangle$.
\item ${\mathcal K}(\mathcal W)\cap {\rm i}{\mathcal K}(\mathcal W)=\{0\}$.
\item ${\mathcal K}(\mathcal W)+{\rm i}{\mathcal K}(\mathcal W)$ is dense in 
${\mathcal H}_1$.\footnote{Such real subspaces of a complex Hilbert space are called {\it 
standard} in the spatial version of Tomita-Takesaki theory \cite{RvD}.
}
\item ${\mathcal K}(\mathcal W)^\perp:=\{\psi\in\mathcal H_1:\,{\rm 
Im}\,\langle\psi,\phi\rangle=0\,\text{for all}\,  
\phi\in {\mathcal K}(\mathcal W)\}={\mathcal K}(\mathcal W')$.
\end{itemize}

The functorial procedure of Weyl quantization (see, e.g., \cite{RS}) now leads for any $\psi\in\bigcup_\mathcal W\mathcal K(\mathcal W)$ to an (unbounded) field operator $\Phi(\psi)$ on the Fock space 
\beq   \mathcal H={\bigoplus}_{n=0}^\infty {\mathcal H}_1^{\otimes_{\rm symm}}\eeq 
 such that
\beq   [\Phi(\psi),\Phi(\phi)]={\rm i}\,{\rm  Im}\, \langle\psi,\phi\rangle\mathbf 1.\eeq {}
In particular,
\beq   [\Phi(\psi_{1}),\Phi(\psi_{2})]=0\eeq 
if $\psi_{1}\in \mathcal K(\mathcal W)$, $\psi_{2}\in \mathcal K(\mathcal W')$.{}

Finally, a net of  algebras $\mathcal F(\mathcal O)$ satisfying requirements R1-R4  is defined by
\beq   \mathcal F(\mathcal O):=\text{weak closure of the algebra generated by\ } \{\exp(\mathrm i\Phi(\psi)):\, \psi\in 
\cap_{\mathcal W\supset \mathcal O}
\mathcal K(\mathcal W)\}\ .\eeq 
These algebras are in \cite{BGL} proved to have $\Omega$ as a cyclic vector if $\mathcal O$ is a space-like cone, i.e, a set of the form $x+\bigcup_{\lambda >0}\lambda D$ where $D$ is a set with interior points that is space-like separated from the origin.

Although this construction produces only interaction free fields it is remarkable for at least two reasons:
{}

\begin{itemize}
\item It uses as sole input a representation of the Poincar\'e group, i.e., it is {\it intrinsically quantum mechanical} and not based on any \lq\lq quantization\rq\rq\   of a classical theory. For the massive representations as well as those of  zero mass and finite helicity the localization can be sharpened by using Wightman fields \cite{SW}, leading to nontrivial algebras 
$\mathcal F(\OO)$ also for {\it bounded}, open sets $\OO$. 
{}
\item The construction works also for  the  {\it zero mass, infinite spin representations}, that can {\it not\/} be generated by point localized fields, i.e., operator valued distributions satisfying the Wightman axioms \cite{Y}. Further analysis of this situation reveals that these representations can be generated by {\it string-localized} fields  $\Phi(x,e)$ \cite{MSY1, Mundetal} with $e$ being a space like vector of length 1, and $[\Phi(x,e), \Phi(x,e)]=0$ if the ``strings" (rays) $x+\lambda e$ and $x'+\lambda' e'$ with $\lambda, \lambda'>0$ are space-like separated. After smearing with test functions in $x$ and $e$ these fields are localized in space-like cones.
\end{itemize}

\noindent{\bf Remarks:}

\begin{itemize}
\item Free string localized fields can be constructed for all irreducible represenations of the Poincar\'e group and their most general form is understood \cite{Mundetal}. Their correlation functions have a better high-energy behavior than  for
Wightman fields. One can also define string localized vector potentials for the electromagnetic field and a massless spin 2 field \cite{Mundpot}, and there are generalizations to massless fields of arbitrary helicities \cite{Plaschke}.
\item
Localization in cones rather than bounded regions occurs also in other contexts:
1) Fields generating massive particle states can always be localized in space-like cones and in massive gauge theories no better localization may be possible \cite{BuchhFred}.\newline
2) In Quantum Electrodynamics localization in {\it light-like} cones is to be expected \cite{BuchhRob}.\end{itemize}

 \section{The Structure of Local Algebras}
 
I recall first some standard mathematical terminology and notations concerning operator algebras.
The algebra of all  bounded, linear operators on  a Hilbert space $\h$ is denoted by $\bh$.
If $\m\subset\bh$ is a subalgebra (more generally, a subset), then its {\it commutant} is, by definition,
\beq   \m'=\{B\in \bh:\, [A,B]=0 \text{ for all } A\in\m\}.\eeq 
A {\it von Neumann algebra} is an algebra $\m$ such that
\beq   \m=\m^{\prime\prime}\ ,\eeq 
i.e., the algebra is equal to its double commutant. Equivalently, the algebra is closed in the weak operator topology, provided the algebra contains ${\mathbf 1}$ that will always be assumed.  A (normal) {\it state} on a von Neumann algebra $\m$ is a {\it positive linear functional} of the form
$\omega(A)={\rm trace}\, (\rho A)$, $\rho\geq 0$, ${\rm trace}\, \rho=1$. The state is a
{\it pure state} if $\omega=\half \omega_1+\half\omega_2$ implies $\omega_1=\omega_2=\omega$. Note that if $\m\neq \bh$ then  $\rho$ is not unique and  the concept of a pure state is {\it not} the same as that of a {\it vector state}, i.e. $\omega(A)=\langle \psi, A\psi\rangle$ with $\psi\in\h$, $\Vert\psi\Vert=1$.

A vector $\psi\in\h$ is called {\it cyclic} for $\m$ if $\m\psi$ is dense in $\h$ and {\it separating} if  $A\psi=0$ for $A\in \m$ implies $A=0$ (equivalently, if $\psi$ is cyclic for $\m'$).


A {\it factor} is a v.N. algebra $\m$ such that
\beq   \m\vee \m':= \left(\m\cup\m'\right)^{''}=\bh\eeq 
which is equivalent to
\beq   \m\cap\m'=\C{\bf 1}.\eeq

The original motivation of von Neumann for introducing and studying this concept together with 
F.J. Murray \cite{MN, MN2, MN3, MN4} came from from quantum mechanics:  A ``factorization" of $\bh$ correponds to a splitting of a system into two subsystems.

The simplest case is
\beq   \h=\h_{1}\otimes\h_{2},\eeq 
\beq   \m=\bho\otimes{\bf 1}\, ,\quad \m'={\bf 1}\otimes\bht,\eeq 
\beq   \bh=\bho\otimes\bht.\eeq
 This is the {\bf Type I} case, familiar from non-relativistic quantum mechanics of systems with a finite number of particles and also Quantum Information Theory \cite{peres} where the Hilbert spaces considered are usually finite dimensional. This factorization is characterized by the existence 
 of {\it minimal projectors} in $\m$:
If $\psi\in \h_{1}$ and  $E_{\psi}=|\psi\rangle\langle\psi|$, then
 \beq   E=E_{\psi}\otimes{\bf 1}\in\m\eeq 
 is a minimal projector, i.e., it has no proper subprojectors in $\m$.

   At the other extreme is the {\bf Type III} case which is defined as follows:
  
  For every projector $E\in\m$ there exists an isometry $W\in \m$ with
  \beq  \label{19} W^*W={\bf 1}\, ,\quad WW^*=E.\eeq
    
  It is clear that for a type III factor, $\m\vee\m'$ is {\it not} a tensor product factorization, because a minimal projector cannot satisfy \eqref{19}.\\

 It is natural to ask whether we need to bother about other cases than type I in quantum physics. The answer is that is simply a {\it fact} that in LQP the algebras  $\F(\mathcal O)$ for $\OO$ a double cone (intersection of a forward and a backward light cone) or a space like wedge, are in all known cases of type III. More precisely, under some reasonable assumptions,  they are isomorphic to the {\it unique}, hyperfinite type  III$_1$ factor in a finer classification due to A. Connes \cite{Connes, Haagerup, BDFred}. This classification is in terms of the intersection of the spectra of the modular operators for the cyclic and separating state vectors for the algebra ({\it Connes spectrum}). The characteristic of type III$_1$  is that the Connes spectrum is equal to $\mathbb R_+$.  
  
 Concrete example of type III factors can be obtained by considering  infinite tensor products of $2\times 2$ or $3\times 3$ matrix algebras \cite{Powers, ArakiWoods, ArakiWoods1}.
Thus, a type $\mathbf {\rm III}_\lambda$ factor with $0<\lambda<1$, which has the integral powers of $\lambda$ as Connes spectrum, is generated by the infinite tensorial power of the algebra $M_2(\mathbb C)$ of complex $2\times 2$ matrices in the Gelfand-Naimark-Segal representation \cite{BR} defined by the state \beq  \omega(A_1\otimes A_2\otimes\cdots \otimes A_N\otimes {\mathbf 1}\cdots)= \prod_n {\rm tr}(A_n \rho_{\lambda})\eeq
with $A_n\in M_2(\mathbb C)$, $0<\lambda<1$ and  \beq  \rho_{\lambda}=\frac 1{1+\lambda}\begin{pmatrix} 1 & 0 \\ 0 & \lambda  \end{pmatrix}.\eeq

A type {$\mathbf {\rm III}_1$ factor is obtained from an analogous formula for the infinite product of complex $3\times 3$ matrices in the representation defined  by tracing with the matrix \beq  \rho_{\lambda, \mu}=\frac 1{1+\lambda+\mu}\begin{pmatrix} 1 & 0 & 0\\ 0 & \lambda & 0\\0 & 0 & \mu \end{pmatrix}\eeq
where $\lambda,\mu>0$ are such that $\frac {\log \lambda}{\log \mu}\notin \mathbb Q$.

The earliest proof of the occurrence of type III factors in LQP was given by H. Araki in \cite{Araki1, Araki2, Araki3} for the case of a free, scalar field. Type III factors appear also in non-relativistic equilibrium quantum statistical mechanics in the thermodynamic limit at nonzero temperature \cite{ArakiWoods1}. 

General proofs that the local algebras of a relativistic quantum field in the vacuum representation are of type III$_1$ rely on the following ingredients:\footnote{The hyperfiniteness, i.e., the approximability by finite dimensional matrix algebras, follows from the split property considered in Sect. 5.1.}

\begin{itemize}
\item The Reeh-Schlieder Theorem. 

\item The Bisognano Wichmann Theorem for the wedge algebras $\F(\mathcal W)$, that identifies their Tomita-Takesaki modular groups w.r.t.\ the vacuum with geometric transformations (Lorentz-boosts). The corresponding modular operators have $\mathbb R_+$ as spectrum. Moreover, by locality and the invariance of the wedge under dilations the spectrum is the same for other cyclic and separating vectors \cite{Araki72, Fredenhag}. Hence the wedge algebra is of type III$_1$.

\item Assumptions about non-triviality of scaling limits that allows to carry the arguments for wedge algebras over to double cone algebras \cite{Fredenhag, BV}.
\end{itemize}

See also \cite{Driessler1, Driessler2} and \cite{Borchers98} for other aspects of the type question.

\subsection{Some consequences of the type III property}

We now collect some important facts about local algebras that follow from their type III character. Since the focus is on the observables, we state the results in terms of the algebras $\m(\OO)$ rather than $\F(\OO)$.

\subsubsection{Local preparability of states:}

For every projector $E\in\mo$ there is an isometry $W\in\mo$ such that if $\omega$ is any state and $\omega_W(\,\cdot\,):= \omega(W^*\,\cdot\,W)$, then
\beq   \omega_W(E)=1,\eeq 
but
\beq  \label{compl}\omega_W(B)=\omega(B)\quad\text{for } B\in\m(\mathcal O').\eeq

In words: Every state can be changed into an eigenstate of a local projector, by a local operation that is independent of the state and does not affect the state in the causal complement of the localization region of the projector. 

This result is a direct consequence of the type III property. It is worth noting that in a slightly weaker form it can be derived from the general assumptions of LQP, without recourse to the Bisognano-Wichmann theorem and the scaling assumptions mentioned above: In \cite{BoM} H.J. Borchers proved that the isometry $W$ can in any case be found in an algebra $\m(\OO_\varepsilon)$ with 
\beq   \label{Oe}\OO_\varepsilon:=\OO+\{x:\, |x|<\varepsilon\},\quad \varepsilon>0.\eeq Eq. \eqref{compl} is then only claimed for $B\in \m(\mathcal O_\varepsilon')$, of course.

In Section \ref{prep} we shall consider a strengthened version of the local preparability, but again with $W\in \m(\OO_\varepsilon)$, under a further assumption on the local algebras.

\subsubsection{Absence of pure states}

  A type III factor $\m$ has {\it no pure states}\footnote{Recall that ``state" means here alway {\em normal} state, i.e. a positive linear functional given by a density matrix in the Hilbert space where $\m$ operates. As a $C^*$ algebra $\m$ has pure states, but these correspond to disjoint representations on different (non separable) Hilbert spaces.}, i.e., for every $\omega$ there are $\omega_{1}$ and $\omega_{2}$, 
{\it different} from $\omega$, such that
\beq   \omega(A)=\mfr1/2\omega_{1}(A)+\mfr1/2\omega_{2}(A)\eeq 
for all $A\in\m$. This means that for local algebras it is not meaningful to interpret statistical mixtures as ``classical" probability distributions superimposed on pure states having a different ``quantum mechanical" probability interpretation,  as sometimes done in textbooks on non-relativistic quantum mechanics.

On the other hand {\it every state} on $\m$ {\it is a vector state}, 
i.e., for every $\omega$ there is a (non-unique!)\footnote{If the algebra is represented in a ``standard form" in the sense of modular theory the vector can be uniquely fixed by taking it from the corresponding ``positive cone" \cite{BR}.}$\psi_{\omega}\in \h$ such that
\beq   \omega(A)=\langle \psi_{\omega},A\psi_{\omega}\rangle\eeq 
for all $A\in\m$ (\cite{Takesaki}, Cor. 3.2, p. 336).

A type III factors has these mathematical features of common with the abelian von Neumann algebra $L^\infty(\mathbb R)$ that also has no pure states whereas every state is a vector state in the natural representation on $L^2(\mathbb R)$. But while $L^\infty(\mathbb R)$ is decomposable into a direct integral of trivial, one-dimensional algebras,  type III factors are noncommutative,  indecomposable and of infinite dimension.


\subsubsection{Local comparison of states cannot be achieved by means of positive operators}

For $\mathcal O$ a subset of Minkowski space and two states  $\varphi$ and $\omega$  we define their {\it local difference} by
\beq \label{locdiff}  D_{\mathcal O}( \varphi, \omega)\equiv\sup\{|\varphi(A)-\omega(A)|\, : A\in\mo, \Vert A\Vert\leq 1\}.\eeq

If $\mathcal A(\mathcal O)$ were a type I algebra local differences could, for a dense set of states, be tested  by means of  positive operators in the following sense:

For a dense set of states $\varphi$ there is a positive operator $P_{\varphi,\mathcal O}$ such that
\beq  \label{posop} D_{\mathcal O}( \varphi, \omega)=0\quad\text{if and only if}\quad \omega(P_{\varphi,\mathcal O})=0.\eeq {}
For a type III algebra, on the other hand,  such operators {\it do not exist for any state}.{}

Failure to  recognize this has in the past led to spurious ``causality problems" \cite{Hegerfeldt}, inferred from the fact that for a positive operator $P$ an expectation value
$\omega(e^{\mathrm iHt}Pe^{-\mathrm iHt})$ with $H\geq 0$ cannot vanish in an interval of $t$'s without vanishing identically.\footnote{This holds because $t\mapsto P^{1/2}e^{-\mathrm iHt}\psi$ is analytic in the complex lower half plane for all vectors $\psi\in\mathcal H$ if $H\geq 0$.} 
In a semi-relativistic model for a gedankenexperiment due to E. Fermi \cite{Fermi} such a positive operator is used to measure the excitation of an atom due to a photon emitted from another atom some distance away. Relativistic causality requires that the excitation takes place only after a time span $t\geq r/c$ where $r$ is the distance between the atoms and $c$ the velocity of light. However, due to the mathematical fact mentioned, an alleged excitation measured by means of a positive operator is in conflict with this requirement. This is not a problem for relativistic quantum field theory, however, where the local difference of states defined by \eqref{locdiff} shows perfectly causal behavior, while a positive operator measuring the excitation simply does not exist \cite{BuchhYng}.\footnote{Already the Corollary to the Reeh-Schlieder Theorem in Sec. 3.3. implies that excitation cannot be measured by a local positive operator since the expectation value of such an operator cannot be zero in a state with bounded energy spectrum. The nonexistence of {\it any} positive operator satisfying \eqref{posop} is a stronger statement.} 

\subsubsection{Remarks on the use of approximate theories}

The above discussion of ``causality problems" prompts the following remarks:

 Constructions of fully relativistic models for various phenomena  where interactions play a decisive role are usually very hard  to carry out in practice. Hence one must as a rule be content with some approximations (e.g., divergent perturbation series without estimates of error terms), or semi-relativistic models with various cut-offs (usually at high energies). Such models usually violate one or more of the general assumptions underlying LQP.
  Computations based on such models may well lead to results that are in conflict with basic principles of relativistic quantum physics such as an upper limit for the propagation speed of causal influence, but this is quite natural and should not be a cause of worry (or of unfounded claims) once the reason is understood.

\section{Entanglement in LQP}

If $\m_1$ and $\m_2$ commute, a  state $\omega$ on $\m_1\vee\m_2$ is by definition {\it entangled}, if it can {\it not} be approximated by convex combinations of product states.{}

Entangled states are ubiquitous in LQP due to the following general mathematical fact \cite{Clifton}:
If $\m_1$ and $\m_2$ commute,
are nonabelian, possess each a cyclic vector and $\m_1\vee\m_2$ has a a separating vector,
then the entangled states form a {\it dense, open}  subset of the set of all states. 

This applies directly to the local algebras of LQP because of the
Reeh-Schlieder Theorem.
Thus the entangled states on $\m(\mathcal O_1)\vee \m(\mathcal O_2)$ are generic for space-like separated, bounded open sets $\mathcal O_1$ and $\mathcal O_2$.


The type III property implies even stronger entanglement:

If $\m$ is a type III factor, then $\m\vee\m'$ does not have {\it any} (normal) product states, i.e.,  {\it all} states are entangled for the pair $\m$, $\m'$. {}

{\it Haag duality} means by definition that $\mo'=\m(\mathcal O')$. Thus, if Haag duality holds,
 a quantum field in a bounded space-time region can never be disentangled from the field in the causal complement.

By allowing a small distance between the regions, however, disentanglement {\it is possible}, provided the theory has a certain property that mitigates to some extent the rigid coupling between a local region and its causal complement implied by the type III character of the local algebras. This will be discussed next.


\subsection{Causal Independence and Split Property}\label{prep}

A pair of commuting 
von Neumann algebras, $\m_{1}$ and $\m_{2}$, in a 
common $\bh$ is
{\it causally (statistically) independent} if for every pair of states, $\omega_{1}$ on 
$\m_{1}$ and $\omega_{2}$ on 
$\m_{2}$, there is a state $\omega$ on $\m_{1}\vee\m_{2}$ such that
\beq   \omega(AB)=\omega_{1}(A)\omega_{2}(B)\eeq 
for $A\in \m_{1}$, $B\in \m_{2}$.{}

In other words: States can be {\it independently} prescribed 
on  $\m_{1}$ and $\m_{2}$ and extended to a common, {\it 
uncorrelated} state on the joint algebra. 
This is really the von Neumann 
concept of independent systems.


 The {\it split property} \cite{DoplicherLongo} for commuting algebras $\m_1$, $\m_2$ means that there is a type I factor $\mathcal N$ such that
\beq   \m_{1}\subset \mathcal N\subset \m_{2}'\eeq 
which again means: There is a tensor product decomposition
$\h=\h_{1}\otimes\h_{2}$ such that
\beq   \m_{1}\subset\bho\otimes{\bf 1}\, ,\quad \m_{2}\subset
{\bf 1}\otimes\bht.\eeq 
In the field theoretic context causal independence and split property 
are equivalent \cite{Buchh74, Wern, Summers}. 

The split property for local algebras separated by a finite distance can be derived from a condition ({\it nuclearity}) that expresses  
the idea that the {\it local energy level density} (measured 
in a suitable sense) does not increase too fast with the energy \cite{BuchhWich, BALong1, BALong2, BuchhYng2}.
Nuclearity is not fulfilled in all models (some generalized free fields provide counterexamples), but it is still a reasonable requirement.

The split property together with the type III property of the 
strictly local algebras leads to a {\it strong version of the local 
preparability of states.} The following result is essentially contained in \cite{BDL}. See also \cite{Wern, Summers}.
\begin{thm}[{\bf Strong local preparability}]  

For
every state $\varphi$ (``target state")  and every bounded open region $\mathcal O$ there is an
isometry $W\in\m({\mathcal O}_{\varepsilon})$ (with ${\mathcal
O}_{\varepsilon}$ slightly larger than $\mathcal O$, cf. \eqref{Oe}) such that for an arbitrary
state $\omega$ (``input state")
\beq  \omega_W(AB)=\varphi(A)\omega (B)\eeq
for $A\in\mo$,  $B\in\m({\mathcal O}_{\varepsilon})'$, where $\omega_W(\,\cdot\,)\equiv\omega(W^*\,\cdot\,W)$.\end{thm}

In particular, $\omega_W$ is uncorrelated and its restriction to $\mo$ is the target state $\varphi$, while in the causal complement of $\mathcal O_\varepsilon$ the preparation has no effect on the input state $\omega$.
Moreower, $W$ depends {\em only on the target state} and not on the input state.

\begin{proof} The split property implies that we can write 
$\mo\subset\bho\otimes{\bf 1}$, $\m({\mathcal 
O}_{\varepsilon})'\subset {\bf 1}\otimes\bht$.\smallskip

By the type III property of $\mo$ we have
$\varphi(A)=\langle\xi,A\xi\rangle$ for $A\in\mo$, with 
$\xi=\xi_{1}\otimes\xi_{2}$. (The latter because every state on a type III factor is a vector state, and we may regard $\mo$ as a type III factor contained in $\mathcal B(\mathcal H_1)$.)
Then $E:=E_{\xi_{1}}\otimes {\bf 1}\in\ \m({\mathcal 
O}_{\varepsilon})''=
\m({\mathcal O}_{\varepsilon})$.\smallskip

 By the type III property of 
$\m({\mathcal O}_{\varepsilon})$ 
there is a 
$W\in\m({\mathcal O}_{\varepsilon})$ with $WW^*=E$, $W^*W={\bf 1}$. The second equality implies
$\omega(W^*BW)=\omega (B)$ 
for $B\in\m({\mathcal O}_{\varepsilon})'.$\smallskip

On the other hand, $EAE=\varphi(A)E$ for $A\in\mo$ and multiplying this equation from left with $W^*$ and right with $W$, one obtains \beq  W^*AW=\varphi(A){\bf 1}\eeq by employing $E=WW^*$and $W^*W={\bf 1}$. Hence \beq  \omega(W^*ABW)=\omega(W^*AWB)=\varphi(A)\omega(B).\eeq
\end{proof}

The Theorem implies also that {\it any state} on 
$\m(\mathcal O)\vee \m(\mathcal O_\varepsilon)'$ {\it can be disentangled by a local operation} in $\m(\mathcal O_\varepsilon)$:  \medskip

Given a state $\omega$ on $\m(\mathcal O)\vee \m(\mathcal O_\varepsilon)'$ there is, by the preceding Theorem, an isometry $W\in \m(\mathcal O_\epsilon)$ such that 
\beq  \omega_W(AB)=\omega(A)\omega(B)\ .\eeq
for all $A\in \m(\OO)$, $B\in\m(\OO_\varepsilon')$.\medskip

In particular: Leaving a security margin between a bounded domain and its causal complement,  the  global vacuum state $\omega(\cdot)=\langle \Omega, \cdot\ \Omega\rangle$, that, being cyclic and separating for the local algebras is entangled between $\m(\mathcal O)$ and any $\mathcal A(\tilde\OO)\subset \m(\mathcal O_\varepsilon)'$  \cite{Clifton, Narnh}, 
can be disentangled by a local operation producing an uncorrelated state on $\m(\mathcal O)\vee\m(\mathcal O_\varepsilon)'$ that is identical to the vacuum state on each of the factors.

\subsection{Conclusions}

Here are some lessons that can be drawn from this brief survey of relativistic quantum physics:

\begin{itemize}
\item LQP provides a framework which resolves the apparent paradoxa
   resulting from combining the particle picture of quantum mechanics
   with special relativity. This resolution achieved by regarding the system as composed of {\it quantum fields} in space-time, represented by a net of local algebras. A subsystem is represented by one of the local algebras, i.e., the fields in a specified part of space-time.  ``Particle" is a derived concept that  emerges asymptotically at large times but (for theories with interactions) is usually not strictly defined at finite times. Irrespective of interactions particle states can never be created by operators strictly localized in bounded regions of space-time.
\item The fact that local algebras have no pure states is relevant for interpretations of the state concept.
It lends support to the interpretation that a state refers to a (real or imagined) ensemble of identically prepared copies of the system but makes it hard to maintain that it is an inherent attribute of an individual copy.

\item The type III property is relevant for causality issues and local preparability of states, and responsible for ``deeply entrenched" entanglement of states between bounded regions and their causal complements, that is, however, mitigated by the split property.
\end{itemize}


On the other hand, the framework of LQP does not {\it per se} resolve all ``riddles" of quantum physics. Those who are puzzled by the violation of Bell's inequalities in EPR type experiments will not necessarily by enlightened by learning that local algebras are type III.
Moreover, the terminology has still an anthropocentric ring (``observables", ``operations") as usual in Quantum Mechanics.
This is disturbing since physics is concerned with more than designed experiments in laboratories. We use quantum (field) theories to understand processes in the interior of stars, in remote galaxies billions of years ago, or even the ``quantum fluctuations" that are allegedly responsible for fine irregularities in the 3K background radiation. In none of these cases ``observers"  were/are around to ``prepare states" or ``reduce wave packets"! A fuller understanding of the emergence of macroscopic ``effects" from the microscopic realm,\footnote{See \cite{Bla} for important steps in this direction and \cite{FS} for a thorough analysis of foundational issues of QM.} without invoking ``operations" or ``observations", and possibly a corresponding revision of the vocabulary of quantum physics is still called for.\footnote{Already Max Planck  in his Leiden lecture of 1908 speaks of the ``Emanzipierung von den antrophomorphen Elementen" as a goal, see \cite{MP} p. 49.}

{\small \subsection*{Acknowledgements} I thank the organizers of the Bielefeld workshop, J\"urg Fr\"ohlich and Philippe Blanchard, for the invitation that lead to these notes, Detlev Buchholz for critical comments on the text, Wolfgang L. Reiter for drawing my attention to ref.\ \cite{MP},  and the Austrian Science Fund (FWF) for support under Project P 22929-N16.}

\bibliographystyle{amsplain}

\begin{thebibliography}{99}
{\small
   
    
 
 \bibitem{Alazzawi}    S.~Alazzawi: ``Deformations of Fermionic Quantum Field Theories and Integrable Models", {\it Lett.~Math.~Phys.\/} {\bf 103}, 37--58 (2013). 

\bibitem{Araki1} H.~Araki: ``A lattice of von Neumann algebras 
associated with the quantum theory of a free Bose field", {\em 
J.~Math.~Phys.\/}, {\bf 4}, 1343--1362 (1963).

\bibitem{Araki2} H.~Araki: ``Von Neumann algebras of local observables for 
free scalar fields", {\em J.~Math.~Phys.\/}, {\bf 5}, 1--13 (1964).

\bibitem{Araki3} H.~Araki: ``Type of von Neumann algebra associated 
with free field", {\em Progr. Theor. Phys.\/} {\bf 32}, 956--965 (1964).

\bibitem{Araki72} H. Araki: ``Remarks on spectra of modular operators of von Neumann algebras",  {\em Commun.
Math. Phys.} {\bf 28},  267--277 (1972)

\bibitem{Arakibook} H. Araki: {\it Mathematical Theory of Quantum Fields}, Oxford University Press, 1999.

\bibitem{ArakiWoods1}  H. Araki, E.J. Woods: ``Representations of 
the canonical commutation relations describing a non-relativistic 
infinite free Bose gas", {\em J. Math. Phys.\/} {\bf 4}, 637--662 
(1963).


   \bibitem{ArakiWoods} H. Araki, E.J. Woods: ``A classification of 
   factors", {\em Publ. R.I.M.S., Kyoto Univ.\/} {\bf 4}, 51--130 
(1968).
   

\bibitem{BJM} J.C.A. Barata, C.D. J\"akel, J. Mund: ``The ${\mathcal P}(\varphi)_2$ Model on the de Sitter Space",  {\it arXiv:1311.2905}.

\bibitem{BisWich1} J.J. Bisognano, E.H. Wichmann: ``On the duality 
condition for a Hermitian scalar field", {\em J. Math. Phys.\/} {\bf 
16}, 985--1007 (1975).

\bibitem{BisWich2} J.J. Bisognano, E.H. Wichmann: ``On the duality 
condition for quantum fields", {\em J. Math. Phys.\/} {\bf 
17}, 303--321 (1976).

\bibitem{Bla} P. Blanchard, R. Olkiewicz: ``Decoherence induced transition from quantum to classical dynamics",  {\em Rev.
Math. Phys.} {\bf 15}, 217--244 (2003).

\bibitem{BogoliubovBook} N.N. Bogoliubov, A.A. Lugonov, A.I. Oksak, I.T. Todorov: {\it General principles of quantum field theory}\/, Kluwer, Dordrecht, 1990.


\bibitem{BorchSuper} H.J. Borchers: ``Local Rings and the Connection of Spin with Statistics", {\em Commun. Math. Phys.} {\bf 1}, 281--307 (1965).

\bibitem{BoM} H.J. Borchers: ``A Remark on a Theorem of B. 
Misra", {\em Comm. Math. Phys.\/} {\bf 4}, 315--323 (1967).




\bibitem{Borchers98} H.J. Borchers: "Half-sided Translations and the 
Type of von Neumann Algebras", 
{\em Lett. Math. Phys.\/} {\bf 44}, 283-290 (1998).



\bibitem{BorchMod} H.J. Borchers: ``On revolutionizing quantum field theory with Tomita's modular theory", {\em J.Math.Phys.} {\bf 41}, 3604--3673 (2000).



\bibitem{BorchY1} H.J. Borchers, J. Yngvason: {\it From Quantum Fields to Local von Neumann 
Algebras.\/} 
{\em Rev.\ Math.\ Phys.\/} Special Issue, 15--47 (1992).

\bibitem{BR} O. Bratteli, D.W. Robinson: {\em Operator Algebras and Quantum Statistical Mechanics I}, Springer 1979.

\bibitem{BDF} R. Brunetti, M. Duetsch, K. Fredenhagen: ``Perturbative Algebraic Quantum Field Theory and the Renormalization Groups",  {\em Adv. Theor. Math. Phys.\/} {\bf 13}, 1541--1599 (2009).


\bibitem{BGL} R. Brunetti, D. Guido, R. Longo: ``Modular Localization and Wigner Particles" 
{\em Rev. Math. Phys.\/} {\bf 14},
759--786 (2002).


\bibitem{Buchh74} D. Buchholz: ``Product States for Local Algebras", 
{\em Comm. Math. Phys.\/} {\bf 
36}, 287--304 (1974).

\bibitem{Bsymmbreak} D. Buchholz: ``Gauss' law and the infraparticle problem", {\em Phys. Lett. B\/} {\bf 174}, 331--334 (1986).

\bibitem{BDFred} D. Buchholz, C. D'Antoni, K. Fredenhagen: 
``The 
Universal Structure of Local Algebras", {\em Comm. Math. Phys.\/} 
{\bf 
111}, 123--135 (1987).

\bibitem{BALong1} D. Buchholz,  C. D'Antoni, R. Longo:``Nuclear 
Maps and Modular Structures I", {\em J. Funct. Anal.\/} {\bf 
88}, 233--250 (1990).

\bibitem{BALong2} D. Buchholz,  C. D'Antoni, R. Longo: ``Nuclear 
Maps and Modular Structures II: Applications to Quantum Field 
Theory", {\em 
Comm. Math. Phys.\/} {\bf 
129}, 115--138 (1990).

\bibitem{BDL} D. Buchholz, S. Doplicher, R. Longo: ``On 
Noether's Theorem in Quantum Field Theory", {\em Ann. Phys.\/} {\bf 
170}, 1--17 (1986).


\bibitem{BuchhFred} D. Buchholz, K. Fredenhagen: ``Locality and the structure of particle states", {\em  Commun. Math. Phys.\/} {\bf 84},
1–54 (1982)

\bibitem{BuchhHaag} D. Buchholz, R. Haag: ``The quest for 
understanding in relativistic quantum physics",  {\em J. Math. Phys.\/} 
{\bf 41}, 3674--3697 (2000). 

\bibitem{BLS} D. Buchholz, G. Lechner, S.J. Summers: ``Warped Convolutions, Rieffel Deformations and the Construction of Quantum Field Theories",
{\em Commun. Math. Phys.} {\bf 304}, 95--123 (2011).

\bibitem{BuchhRob} D. Buchholz, J.E. Roberts: ``New Light on Infrared Problems:
Sectors, Statistics, Symmetries and Spectrum", {\em arXiv:1304.2794\/}.

\bibitem{BV} D. Buchholz, R. Verch: "Scaling algebras and 
renormalization group in algebraic quantum field theory", {\em Rev. 
Math. Phys.\/} {\bf 7}, 1195--1239 (1996).


\bibitem{BuchhWich} D. Buchholz, E. Wichmann: ``Causal Independence 
and 
the Energy level Density of States in Local Quantum Field theory", 
{\em 
Comm. Math. Phys.\/} {\bf 
106}, 321--344  (1986).


\bibitem{BuchhYng2} D. Buchholz,  J. Yngvason: ``Generalized Nuclearity 
Conditions and the Split Property in Quantum Field Theory", {\em 
Lett. Math. Phys.\/} {\bf 23}, 159--167 (1991).

\bibitem{BuchhYng} D. Buchholz, J. Yngvason:  ``There Are No Causality 
Problems in Fermi's Two Atom System", {\em Phys. Rev. Lett.\/} {\bf 
73}, 613--613 (1994).

\bibitem{Clifton} R. Clifton, H. Halvorson: ``Entanglement and open systems in algebraic quantum field theory",
{\em Stud. Hist. Philosoph. Mod. Phys.\/} {\bf 32}, 1--31 (2001).

\bibitem{Connes} A. Connes:  ``Une classification des facteurs de 
   type III", {\em Ann. Sci. Ecole Norm. Sup. \/}{\bf 6}, 133--252 
(1973).

\bibitem{DHR1} S. Doplicher, R. Haag and J.E. Roberts: ``Fields, observables and gauge transformations
II", {\em Commun. Math. Phys.} {\bf 15}, 173–200 (1969).

\bibitem{DHR2} S. Doplicher, R. Haag and J.E. Roberts: ``Local observables and particle statistics
I", {\em Commun. Math. Phys.} {\bf 23}, 199–230 (1971).

\bibitem{DHR3} S. Doplicher, R. Haag and J.E. Roberts: ``Local observables and particle statistics
II", {\em Commun. Math. Phys.} {\bf 35}, 49–85 (1974).

\bibitem{DoplicherLongo} S. Doplicher, R. Longo: ``Standard and split 
inclusions of von Neumann algebras", {\em Invent. Math.\/} {\bf 75}, 
493--536 (1984).

\bibitem{DHR4} S. Doplicher and J.E. Roberts: ``Why there is a field algebra with a compact
gauge group describing the superselection structure in particle physics", 
{\em Commun. Math. Phys.} {\bf 131}, 51–107 (1990).



\bibitem{Driessler1} W. Driessler: ``Comments on Lightlike 
Translations 
and Applications in Relativistic Quantum Field Theory", 
{\em Comm. Math. Phys.\/} {\bf 44}, 133--141 (1975).


\bibitem{Driessler2} W. Driessler: ``On the Type of Local Algebras in 
Quantum Field Theory", 
{\em Comm. Math. Phys.\/} {\bf 53}, 295--297 (1977).

\bibitem{Driessleretal} W. Driessler, S.J. Summers, E.H. Wichmann: ``On the connection between quantum fields and von Neumann algebras of local operators",
{\em Comm. Math. Phys.\/} {\bf 105}, 49--84 (1986).

\bibitem{Fermi} E. Fermi: ``Quantum Theory of Radiation", {\em Rev. 
Mod. Phys.\/} {\bf 4},  87--132 (1932).


\bibitem{Fredenhag} K. Fredenhagen: ``On the Modular Structure of 
Local 
Algebras of Observables", {\em Commun. Math. Phys.\/} {\bf 84}, 
79--89 (1985).

\bibitem{Fredclust}  K. Fredenhagen: ``A Remark on the Cluster Theorem", 
{\em Comm. Math. Phys.\/} {\bf 97}, 461--463 (1985).

\bibitem{FS} J. Fr\"ohlich, B. Schubnel: ``Quantum Probability Theory and the Foundations of Quantum mechanics", {\em arXiv:1310.1484v1 [quant-ph].}


\bibitem{Froehlich} J. Fr\"ohlich, G. Morchio, F. Strocchi: ``Infrared Problem and Spontaneous Breaking of the
Lorentz Group in QED", {\em Phys. Lett.  B} {\bf 89} 61--64 (1979).


\bibitem{GlimmJaff} J. Glimm, A. Jaffe: {\em  Quantum Physics: A Functional Integral Point of View\/}, Springer 1987.

\bibitem{GrosseLechner} H. Grosse, G. Lechner: ``Wedge-Local Quantum Fields and Noncommutative Minkowski Space", {\em JHEP} 0711 (2007).

\bibitem{Haagthm}R. Haag: ``On quantum field theories", {\em Mat.-fys. Medd. Kong. Danske Videns. Selskab} {\bf 29}, Nr.12 (1955).  

\bibitem{Haagorg} R. Haag: ``Quantum field theory with composite particles and asymptotic conditions", {\em Phys, Rev.\/} {\bf 112}, 669--673 (1958).

\bibitem{Haag} R. Haag: {\em Local Quantum Physics\/}, 
Springer, Berlin etc 1992.

\bibitem{HHW} R. Haag, N. M. Hugenholtz, and M. Winnink: ``On the equilibrium states in quantum statistical mechanics", 
{\em Comm. Math. Phys.} {\bf 5} 215--236 (1967).


\bibitem{HK} R. Haag, D. Kastler: ``An algebraic approach to quantum field theory“, J. Math. Phys. {\bf 5}, 848--861 (1964). 


   \bibitem{Haagerup} U. Haagerup: ``Connes' bicentralizer problem and
   uniqueness of injective factors of type ${\rm III}_{1}$", {\em Acta
   Math.\/}, {\bf 158}, 95--148 (1987).
   

\bibitem{Hegerfeldt} G.C. Hegerfeldt: ``Causality 
Problems in Fermi's Two Atom System", {\em Phys. Rev. Lett.\/}, {\bf 
72}, 596-599 (1994) .




\bibitem{Jost} R. Jost: {\em The general theory of quantized fields\/},  Am. Math. Soc., Providence, RI 1965.



\bibitem{KawahigashiLongo} Y. Kawahigashi, R. Longo: ``Classification of two-dimensional local conformal nets
with $c < 1$ and 2-cohomology vanishing for tensor categories", {\em Comm. Math. Phys.\/}
{\bf 244}, 63--97 (2004).



\bibitem{Kraus} K. Kraus: {\em States, Effects and Operations}, Springer, Berlin etc 1983.

\bibitem{Lechner} G. Lechner: ``Deformations of quantum field theories and integrable models",  {\em Comm. Math. Phys.} {\bf 312}, 265--302 (2012). 

\bibitem{Longo} R. Longo: ``Notes on Algebraic Invariants for 
Non-commutative Dynamical Systems",  
{\em Comm. Math. Phys.\/} {\bf 69}, 195--207 (1979).

\bibitem{Malament} D.B. Malament: ``In Defence of a Dogma: Why there cannot be a relativistic quantum mechanics of (localizable) particles", in: R. Clifton
(ed.):{\em Perspectives on Quantum Reality\/},  pp. 1--10, Kluwer, Dordrecht 1996.
Montreal 2001.

\bibitem{Mundpot} J.~ Mund: ``String-localized Quantum Fields, Modular Localization and Gauge Theories",
in: V. \ Sidoravicius (Ed.) \emph{New Trends in Mathematical Physics}, pp. 495--508, Springer, 2009.



\bibitem{MSY1} J. Mund, B. Schroer, J. Yngvason: ``String-localized quantum fields from Wigner representations", 
{\em Phys. Lett. B} {\bf 596}, 156–162 (2004).


\bibitem{Mundetal} J. Mund, B. Schroer, J. Yngvason: ``String-Localized Quantum Fields and Modular
Localization",
 {\em Commun. Math. Phys.} {\bf 268}, 621--672 (2006).


    \bibitem{MN} F.J. Murray, J. von Neumann:  ``On Rings of 
   Operators", {\em Ann. Math.\/} {\bf 37}, 116--229 (1936).
  
  \bibitem{MN2} F.J. Murray, J. von Neumann:  ``On Rings of 
   Operators II", {\em Trans. Am. Math. Soc.\/} {\bf 41}, 208--248 
(1937).
   
   \bibitem{MN3} J. von Neumann:  ``On Rings of 
   Operators III", {\em Ann. Math.\/} {\bf 41}, 94--161 (1940).
   
    \bibitem{MN4} F.J. Murray, J. von Neumann:  ``On Rings of 
   Operators IV", {\em Ann. Math.\/} {\bf 44}, 716--808 (1943).
    



\bibitem{Narnh} H. Narnhofer: ``The role of transposition and CPT operation for entanglement", {\em Phys. Lett. A}, 423--433 (2003).

\bibitem{NW} T.D. Newton, E.P. Wigner: ``Localized States for Elementary Systems", {\em Rev. Mod. Phys.} {\bf 21}, 400–406
(1949).


\bibitem{peres} A. Peres, D.R. Terno: ``Quantum information and relativity theory", {\em Rev. Mod. Phys.\/} {\bf 76}, 93--123 (2004).

\bibitem{PerezWilde} J.F. Perez, I.F. Wilde: ``Localization and causality in relativistic quantum mechanics", {\em Phys. Rev D.} {\bf 16}, 315--317 (1977).

\bibitem{MP} M. Planck: {\em Vortr\"age und Erinnerungen}, S. Hirzel Verlag, 1949.

\bibitem{Plaschke} M. Plaschke, J. Yngvason: ``Massless String Fields for Any Helicity", {\em J. Math. Phys.} {\bf 53}, 042301 (2012).

  \bibitem{Powers} R.T. Powers: ``Representations of uniformly 
   hyperfinite algebras and their associated von Neumann rings", {\em 
Ann. 
   Math.\/},{\bf 86}, 138--171 (1968).
   
 
\bibitem{RS} Reed, M., Simon, B.: {\em Methods of modern mathematical physics II}, Academic Press, 1975.



   
   \bibitem{ReehSchlied} H. Reeh, S. Schlieder: ``Eine Bemerkung zur 
Unit\"arequivalenz von Lorentzinvarianten Feldern", {\em Nuovo 
Cimento} {\bf 22}, 1051--1068 (1961).

  \bibitem{RvD} M. A. Rieffel , A. Van Daele: ``A bounded operator approach to Tomita-Takesaki theory", {\em Pacific J. Math.} {\bf 69}, 187-221 (1977).




\bibitem{schwartz}
J.~Schwartz: ``Free Quantized Lorentzian Fields", {\em J. Math. Phys.}
  \textbf{2}, 271--290 (1960).

 \bibitem{SW}
R.~F. Streater and A.~S. Wightman, \emph{PCT, Spin and Statistics, and All
  That}, W. A. Benjamin Inc., New York, 1964.


\bibitem{Summers} S.J. Summers: ``On the independence of local algebras 
in 
quantum field theory", {\em Rev. Math. Phys.\/} {\bf 2}, 201--247 
(1990).

\bibitem{SummMod} S.J. Summers: ``Tomita-Takesaki Modular Theory", {arXiv:math-ph/0511034}.


\bibitem{WS} S.J. Summers,  R.F.  Werner:
``On Bell's Inequalities and Algebraic Invariants", {\em Lett. Math. Phys.} {\bf 33}, 321-334 (1995).

  \bibitem{Tomita} M. Takesaki: ``Tomita's Theory of Modular Hilbert 
Algebras and Its Applications", {\em Lecture Notes in Mathematics} 
{\bf 128}, Springer, Berlin etc 1970.


\bibitem{Takesaki} M. Takesaki: {\em Theory of Operator Algebras II}, Springer, 2003.


\bibitem{Wern} R. F. Werner: ``Local preparability of States and the 
Split Property in Quantum Field Theory", {\em Lett. Math. Phys.\/} 
{\bf 13}, 325--329 (1987).



\bibitem{GaardW} A. S. Wightman and L. G\aa rding: ``Fields as operator-valued distributions in relativistic quantum theory", {\em Arkiv f\o r Fysik} {\bf 28}, 129--189 (1965).


\bibitem{Wig} E. Wigner: ``On Unitary Representations of the Inhomogeneous Lorentz Group", 
{\em Ann. Math.}
Sec. Series, {\bf 40}, 149--204 (1939).

\bibitem{Y} J. Yngvason: ``Zero-mass infinite spin representations of the Poincaré group and quantum field theory",
{\em Commun. Math. Phys.} {\bf 18}, 195-203 (1970).

\bibitem{Zych} M. Zych, F. Costa, J. Kofler, C. Brukner: ``Entanglement between smeared field operators in the Klein-Gordon vacuum", {\em Phys. Rev. D} {\bf 81}, 125019 (2010).

 } \end{thebibliography}

\end{document}